\documentclass{article}%
\usepackage{amssymb}%
\usepackage{amsmath}%
\usepackage{amsfonts}%
\usepackage{graphicx}
\providecommand{\U}[1]{\protect\rule{.1in}{.1in}}
\newtheorem{theorem}{Theorem} [section]

\newtheorem{corollary}[theorem]{Corollary}

\newtheorem{lemma}[theorem]{Lemma}

\newtheorem{proposition}[theorem]{Proposition}
\newtheorem{remark}[theorem]{Remark}

\newenvironment{proof}[1][Proof]{\noindent\textbf{#1.} }{\ \rule{0.5em}{0.5em}}
\begin{document}

\author{Vadim E. Levit\\Department of Computer Science and Mathematics\\Ariel University Center of Samaria, Ariel, Israel\\levitv@ariel.ac.il
\and Eugen Mandrescu\\Department of Computer Science\\Holon Institute of Technology, Holon, Israel\\eugen\_m@hit.ac.il}
\title{When $G^{2}$ is a K\"{o}nig-Egerv\'{a}ry graph?}
\maketitle

\begin{abstract}
The \textit{square} of a graph $G$ is the graph $G^{2}$ with the same vertex
set as in $G$, and an edge of $G^{2}$ is joining two distinct vertices,
whenever the distance between them in $G$ is at most $2$. $G$ is a
square-stable graph if it enjoys the property $\alpha(G)=\alpha(G^{2})$, where
$\alpha(G)$ is the size of a maximum stable set in $G$.

In this paper we show that $G^{2}$ is a K\"{o}nig-Egerv\'{a}ry graph if and
only if $G$ is a square-stable K\"{o}nig-Egerv\'{a}ry graph.\newline


\end{abstract}

\section{Introduction}

\mathstrut All the graphs considered in this paper are finite, undirected,
loopless and without multiple edges. For such a graph $G=(V,E)$ we denote its
vertex set by $V=V(G)$ and its edge set by $E=E(G)$. If $X\subset V$, then
$G[X]$ is the subgraph of $G$ spanned by $X$. By $G-W$ we mean the subgraph
$G[V-W]$, if $W\subset V(G)$.

The neighborhood of a vertex $v\in V$ is the set $N(v)=\left\{  w:w\in V\text{
\ \textit{and} }vw\in E\right\}  $, and $N(A)=\cup\left\{  N(v):v\in A\text{
}\right\}  $, for $A\subset V$. If $\left\vert N(v)\right\vert =\left\vert
\{w\}\right\vert =1$, then $v$ is a \textit{leaf} and $vw$ is a
\textit{pendant edge} of $G$.

By $C_{n}$, $K_{n}$, $P_{n}$ we denote the chordless cycle on $n\geq$ $4$
vertices, the complete graph on $n\geq1$ vertices, and respectively the
chordless path on $n\geq3$ vertices.

A stable set of maximum size will be referred as to a \textit{stability
system} of $G$. The \textit{stability number }of $G$, denoted by $\alpha(G)$,
is the cardinality of a stability system in $G$. Let $\Omega(G)$ denotes
$\{S:S$ \textit{is a stability system of} $G\}$.

A \textit{matching} is a set of non-incident edges of $G$; a matching of
maximum cardinality $\mu(G)$ is a \textit{maximum matching}, and a matching
covering all the vertices of $G$ is called a \textit{perfect matching}. $G$ is
a \textit{K\"{o}nig-Egerv\'{a}ry graph }provided $\alpha(G)+\mu(G)=\left\vert
V(G)\right\vert $, \cite{dem}, \cite{ster}. 

If $S$ is an independent set of a graph $G$ and $H=G[V-S]$, then we write
$G=S\ast H$. Clearly, any graph admits such representations.

\begin{theorem}
\label{th2}\cite{levm4} If $G$ is a graph, then the following assertions are equivalent:

\emph{(i)} $G$ is a \textit{K\"{o}nig-Egerv\'{a}ry} graph;

\emph{(ii)} $G=S\ast H$, where $S\in\Omega(G)$ and $\left\vert S\right\vert
\geq\mu(G)=\left\vert V(H)\right\vert $;

\emph{(iii)} $G=S\ast H$, where $S$ is an independent set with $\left\vert
S\right\vert \geq\left\vert V(H)\right\vert $ and $(S,V(H))$ contains a
matching $M$ of size $\left\vert V(H)\right\vert $.
\end{theorem}

$G$ is \textit{well-covered} if it has no isolated vertices and if every
maximal stable set of $G$ is also a maximum stable set, i.e., it is in
$\Omega(G)$ \cite{plum}. $G$ is called \textit{very well-covered} \cite{fav1},
provided $G$ is well-covered and $\left\vert V(G)\right\vert =2\alpha(G)$.
Some interrelations between well-covered and K\"{o}nig-Egerv\'{a}ry graphs
were studied in \cite{LevMan1998}, \cite{LevMan1999}.

The distance between two vertices $v,w\in V(G)$ is denoted by $dist_{G}(v,w)$,
or $dist(v,w)$ if no ambiguity. $G^{2}$ denotes the second power of graph $G$,
i.e., the graph with the same vertex set $V$ and an edge is joining distinct
vertices $v,w\in V$ whenever $dist_{G}(v,w)\leq2$. Clearly, any stable set of
$G^{2}$ is stable in $G$, as well, while the converse is not generally true.
Therefore, we may assert that $1\leq\alpha(G^{2})\leq\alpha(G)$. Let notice
that the both bounds are sharp. For instance, if:

\begin{itemize}
\item $G$ is not a complete graph and $dist(a,b)\leq2$ holds for any $a,b\in
V(G)$, then $\alpha(G)\geq2>1=\alpha(G^{2})$; e.g., for the $n$-star graph
$G=K_{1,n}$, with $n\geq2$, we have $\alpha(G)=n>$ $\alpha(G^{2})=1;$

\item $G=P_{4}$, then $\alpha(G)=\alpha(G^{2})=2$.
\end{itemize}

The graphs $G$ for which the upper bound of the above inequality is achieved,
i.e., $\alpha(G)=\alpha(G^{2})$, are called \textit{square-stable}; e.g., the
graph from Figure \ref{fig1}. \begin{figure}[h]
\setlength{\unitlength}{1cm}\begin{picture}(5,1.2)\thicklines
\multiput(4,0)(1,0){2}{\circle*{0.29}}
\multiput(4,1)(1,0){3}{\circle*{0.29}}
\put(4,0){\line(1,0){1}}
\put(5,1){\line(1,0){1}}
\put(5,0){\line(0,1){1}}
\put(5,0){\line(1,1){1}}
\put(4,0){\line(0,1){1}}
\multiput(7,0)(1,0){2}{\circle*{0.29}}
\multiput(7,1)(1,0){3}{\circle*{0.29}}
\put(7,0){\line(1,0){1}}
\put(8,1){\line(1,0){1}}
\put(7,0){\line(1,1){1}}
\put(7,0){\line(0,1){1}}
\put(7,1){\line(1,-1){1}}
\put(8,0){\line(0,1){1}}
\put(8,0){\line(1,1){1}}
\put(7,0){\line(2,1){2}}
\end{picture}
\caption{A square-stable graph $G$ and its $G^{2}$.}%
\label{fig1}%
\end{figure}
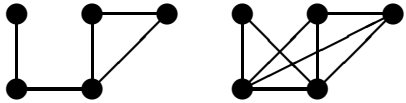

\begin{theorem}
\label{th5}\cite{LevMan2005} The graph $G$ is square-stable if and only if
there is some $S\in\Omega(G)$ such that $dist_{G}(a,b)\geq3$ holds for all
distinct $a,b\in S$.
\end{theorem}

In this paper we prove that $G^{2}$ is a K\"{o}nig-Egerv\'{a}ry graph if and
only if $G$ is a square-stable K\"{o}nig-Egerv\'{a}ry graph. In particular, we
deduce that the square of the tree $T$ is a K\"{o}nig-Egerv\'{a}ry graph if
and only if $T$ is well-covered.

\section{Results}

It is quite evident that $G$ and $G^{2}$ are simultaneously connected or
disconnected. Thus in the rest of the paper all the graphs are connected.

\begin{lemma}
\label{Lemma1} If $G$ is a square-stable graph with $2$ vertices at least,
then $\alpha(G)\leq\mu(G)$.
\end{lemma}

\begin{proof}
According to Theorem \ref{th5}\emph{ }there exists a maximum stable set
\[
S=\left\{  v_{i}:1\leq i\leq\alpha(G)\right\}
\]
in $G$ such that $dist_{G}\left(  a,b\right)  \geq3$ for all pairwise distinct
$a,b\in S$. It follows that for every $i\in\left\{  1,2,...,\alpha\left(
G\right)  -1\right\}  $ there is a shortest path in $G$, of length $3$ at
least, connecting $v_{i}$ to $v_{\alpha\left(  G\right)  }$, say $v_{i}%
,w_{i},...w^{i},$ $v_{\alpha\left(  G\right)  }$ (see Figure \ref{Fig2}).
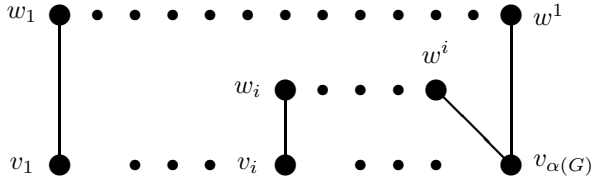
\begin{figure}[h]
\setlength{\unitlength}{1cm}\begin{picture}(5,2.3)\thicklines
\put(4,0){\circle*{0.29}}
\put(4,2){\circle*{0.29}}
\put(10,2){\circle*{0.29}}
\put(4,0){\line(0,1){2}}
\put(10,0){\line(0,1){2}}
\multiput(4,2)(0.5,0){12}{\circle*{0.15}}
\multiput(5,0)(0.5,0){3}{\circle*{0.15}}
\put(7,0){\circle*{0.29}}
\put(7,1){\circle*{0.29}}
\put(7,0){\line(0,1){1}}
\multiput(7,1)(0.5,0){4}{\circle*{0.15}}
\put(9,1){\circle*{0.29}}
\put(9,1){\line(1,-1){1}}
\multiput(8,0)(0.5,0){3}{\circle*{0.15}}
\put(10,0){\circle*{0.29}}
\put(3.5,0){\makebox(0,0){$v_{1}$}}
\put(3.5,2){\makebox(0,0){$w_{1}$}}
\put(6.5,0){\makebox(0,0){$v_{i}$}}
\put(6.5,1){\makebox(0,0){$w_{i}$}}
\put(9,1.5){\makebox(0,0){$w^{i}$}}
\put(10.5,2){\makebox(0,0){$w^{1}$}}
\put(10.7,0){\makebox(0,0){$v_{\alpha (G)}$}}
\end{picture}
\caption{$S=\{v_{1},...,v_{i},...,v_{\alpha(G)}\}\in\Omega(G)$ and
$M=\{v_{1}w_{1},...,v_{i}w_{i},...,v_{\alpha(G)}w^{1}\}$ is a matching in
$G$.}%
\label{Fig2}%
\end{figure}

All the vertices $w_{i},1\leq i\leq\alpha\left(  G\right)  -1$ and $w^{1}$ are
pairwise distinct, i.e.,
\[
w_{i}\neq w^{1},1\leq i\leq\alpha\left(  G\right)  -1,
\]
because, otherwise, there will be a pair of vertices in $S$ at distance $2$,
in contradiction with the hypothesis on $S$. Hence we deduce that
\[
M=\left\{  v_{i}w_{i}:1\leq i\leq\alpha(G)-1\right\}  \cup\left\{
v_{\alpha(G)}w^{1}\right\}
\]
is a matching in $G$ that saturates all the vertices of $S\in\Omega(G)$.
Consequently, we obtain $\alpha(G)=\left\vert S\right\vert =\left\vert
M\right\vert \leq\mu(G)$.
\end{proof}

\begin{remark}
The vertex $w^{1}$ in the proof of Lemma \ref{Lemma1} may be a common vertex
for more shortest paths connecting various $v_{i}$ to $v_{\alpha\left(
G\right)  }$ (see Figure \ref{fig11}).
\end{remark}

\begin{figure}[h]
\setlength{\unitlength}{1cm}\begin{picture}(5,2)\thicklines
\multiput(4,0.5)(1,0){5}{\circle*{0.29}}
\multiput(4,1.5)(2,0){3}{\circle*{0.29}}
\put(9,1.5){\circle*{0.29}}
\put(4,0.5){\line(1,0){4}}
\multiput(4,0.5)(2,0){3}{\line(0,1){1}}
\multiput(4,1.5)(2,0){2}{\line(1,-1){1}}
\put(8,1.5){\line(1,0){1}}
\put(8,0.5){\line(1,1){1}}
\put(4,1.9){\makebox(0,0){$v_{1}$}}
\put(6,1.9){\makebox(0,0){$v_{2}$}}
\put(9.4,1.5){\makebox(0,0){$v_{3}$}}
\put(5,0){\makebox(0,0){$w_{1}$}}
\put(7,0){\makebox(0,0){$w_{2}$}}
\put(8,0){\makebox(0,0){$u$}}
\put(3,1){\makebox(0,0){$G$}}
\end{picture}
\caption{$G$ has $\alpha(G)=\alpha(G^{2})=3=|\{v_{1}w_{1},v_{2}w_{2}%
,v_{3}u\}|<\mu(G)$, where $w^{1}=w^{2}=u$. }%
\label{fig11}%
\end{figure}
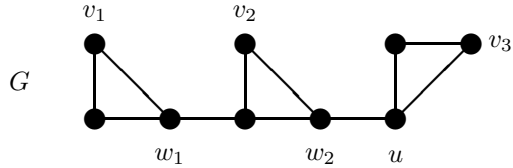

The graph $G$ in Figure \ref{fig1} is square-stable and has $\mu(G)=\mu
(G^{2})=2$, while the square-stable graph $G$ from Figure \ref{fig2} satisfies
$\mu(G)<\mu(G^{2})$. Notice that, in the both examples, neither $G$ nor
$G^{2}$ is a K\"{o}nig-Egerv\'{a}ry graph. 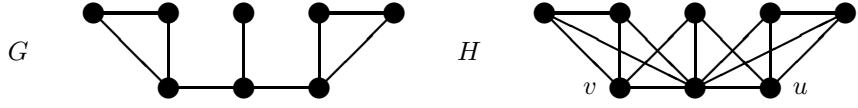
\begin{figure}[h]
\setlength{\unitlength}{1cm}\begin{picture}(5,1.2)\thicklines
\multiput(3,0)(1,0){3}{\circle*{0.29}}
\multiput(2,1)(1,0){5}{\circle*{0.29}}
\put(3,0){\line(1,0){2}}
\multiput(3,0)(1,0){3}{\line(0,1){1}}
\multiput(2,1)(3,0){2}{\line(1,0){1}}
\put(2,1){\line(1,-1){1}}
\put(5,0){\line(1,1){1}}
\put(1,0.5){\makebox(0,0){$G$}}
\multiput(9,0)(1,0){3}{\circle*{0.29}}
\multiput(8,1)(1,0){5}{\circle*{0.29}}
\put(9,0){\line(1,0){2}}
\put(9,0){\line(1,1){1}}
\put(9,1){\line(1,-1){1}}
\multiput(9,0)(1,0){3}{\line(0,1){1}}
\multiput(8,1)(3,0){2}{\line(1,0){1}}
\put(8,1){\line(1,-1){1}}
\put(8,1){\line(2,-1){2}}
\put(10,0){\line(2,1){2}}
\put(10,0){\line(1,1){1}}
\put(10,1){\line(1,-1){1}}
\put(11,0){\line(1,1){1}}
\put(8.6,0){\makebox(0,0){$v$}}
\put(11.4,0){\makebox(0,0){$u$}}
\put(7,0.5){\makebox(0,0){$H$}}
\end{picture}
\caption{$G^{2}=H+vu$ has $\alpha(G^{2})=\alpha(G)$, while $\mu(G)<\mu(G^{2}%
)$. }%
\label{fig2}%
\end{figure}

\begin{proposition}
\label{th1}Let $G^{2}$ be a K\"{o}nig-Egerv\'{a}ry graph with $2$ vertices at
least. Then the following assertions are equivalent:

\emph{(i)} $\alpha(G)=\alpha(G^{2})$;

\emph{(ii)} $\mu(G)=\mu(G^{2})$;

\emph{(iii)} $G$ is a K\"{o}nig-Egerv\'{a}ry graph with a perfect matching.
\end{proposition}

\begin{proof}
The following inequalities are true for every graph $G$:
\[
\mu(G)\leq\mu(G^{2})\ \text{\textit{and}}\ \alpha(G^{2})\leq\alpha(G).
\]

Since $G^{2}$ is a K\"{o}nig-Egerv\'{a}ry graph, $\mu(G^{2})\leq\alpha\left(
G^{2}\right)  $. Consequently, we get
\[
\mu(G)\leq\mu(G^{2})\leq\alpha(G^{2})\leq\alpha(G).
\]

\emph{(i) }$\Longrightarrow$ \emph{(ii),(iii) }If $G$ is square-stable, then
these inequalities together with Lemma \ref{Lemma1}\ give
\[
\mu(G)=\mu(G^{2})=\alpha(G^{2})=\alpha(G).
\]
Moreover, we infer that
\[
\left\vert V(G)\right\vert =\mu(G^{2})+\alpha(G^{2})=\mu(G)+\alpha(G),
\]
which means that $G$ is a K\"{o}nig-Egerv\'{a}ry graph. In addition, $G$\ has
a perfect matching, because $\mu(G)=\alpha\left(  G\right)  $.

\emph{(iii) }$\Longrightarrow$ \emph{(i) }If $G$ is a K\"{o}nig-Egerv\'{a}ry
graph with a perfect matching, then
\[
\mu(G)+\alpha(G)=\left\vert V(G)\right\vert =\mu(G^{2})+\alpha(G^{2}%
)\ \text{\textit{and}}\ \mu(G)=\mu(G^{2}).
\]
Thus, we deduce that $\alpha(G)=\alpha(G^{2})$, i.e., $G$ is a square-stable graph.

\emph{(ii) }$\Longrightarrow$ \emph{(i) }If $\mu(G)=\mu(G^{2})$, then it
follows that
\[
\left\vert V(G)\right\vert =\alpha(G^{2})+\mu(G^{2})\leq\alpha(G)+\mu
(G^{2})=\alpha(G)+\mu(G)\leq\left\vert V(G)\right\vert ,
\]
which assures that $\alpha(G)=\alpha(G^{2})$, i.e., $G$ is a square-stable graph.
\end{proof}

It is worth noticing that if $G$ is square-stable, then it is not enough to
know that $\mu(G)=\alpha(G)$ in order to be sure that $G$ is a
K\"{o}nig-Egerv\'{a}ry graph; e.g., the graph from Figure \ref{fig1}.

\begin{remark}
There are K\"{o}nig-Egerv\'{a}ry graphs, whose squares are not
K\"{o}nig-Egerv\'{a}ry graphs; e.g., every even chordless cycle.
\end{remark}

\begin{remark}
There are non-K\"{o}nig-Egerv\'{a}ry graphs, whose squares are not
K\"{o}nig-Egerv\'{a}ry graphs; e.g., every odd chordless cycle.
\end{remark}

\begin{theorem}
\label{th3}If $G^{2}$ is a K\"{o}nig-Egerv\'{a}ry graph, then $G$ is a
square-stable K\"{o}nig-Egerv\'{a}ry graph with a perfect matching.
\end{theorem}

\begin{proof}
Since $G^{2}$ is a K\"{o}nig-Egerv\'{a}ry graph, Theorem \ref{th2} ensures
that $G^{2}=S\ast H$, where $S\in\Omega(G^{2})$, $\mu(G^{2})=\left\vert
V(H)\right\vert $ and every maximum matching of $G^{2}$ is contained in
$(S,V(H))$.

Let $S=\{s_{j}:1\leq j\leq\alpha(G^{2})\}\in\Omega(G^{2})$ and $V(H)=\{h_{k}%
:1\leq k\leq\left\vert V(G)\right\vert -\alpha(G^{2})\}$.

\textit{Claim 1.} Every $h\in V(H)$ is joined, by an edge from $G$, to at most
one vertex of $S$.

Otherwise, if some $h\in V(H)$ has two neighbors $s_{i},s_{j}\in S$ such that
$hs_{i},hs_{j}\in E(G)$, then $s_{i}s_{j}\in E(G^{2})$, in contradiction to
the fact that $S$ is independent.

\textit{Claim 2.} $S_{G}(H)=S_{G^{2}}(H)$, where
\begin{align*}
S_{G}(H)  & =\{s\in S:\left(  \exists\right)  hs\in E(G),h\in V(H)\},\text{
and}\\
S_{G^{2}}(H)  & =\{s\in S:\left(  \exists\right)  hs\in E(G^{2}),h\in V(H)\}.
\end{align*}

Since $E(G)\subseteq E(G^{2})$, we get that $S_{G}(H)\subseteq S_{G^{2}}(H)$.
Assume that there is some $s\in S_{G^{2}}(H)-S_{G}(H)$. Hence, it follows that
there is some $h_{j}s\in E(G^{2})-E(G)$. Consequently, in $G$ must exist some
path on two edges from $s$ to $h_{j}$, and because $S$ is stable, it follows
that there is some $h_{k}\in V(H)$, such that $h_{k}h_{j},h_{k}s\in E(G)$ and
this contradicts the fact that $s\in S_{G^{2}}(H)-S_{G}(H)$.

\textit{Claim 3.} There is a maximum matching in $G^{2}$ containing only edges
from $G$.

Combining \textit{Claim} 1 and \textit{Claim} 2, it follows that every $h\in
V(H)$ is joined, by an edge from $G$, to exactly one vertex of $S$, say $s(h)
$, because, otherwise, we get $S_{G}(H)\neq S_{G^{2}}(H)$. Now, the set
$M=\{hs(h):h\in V(H)\}$ is a matching both in $G$ and in $G^{2}$. Moreover, by
Theorem \ref{th2}, $M$ is a maximum matching in $G^{2}$, because $\left\vert
M\right\vert =\left\vert V(H)\right\vert $. Consequently, we deduce that
$\left\vert M\right\vert \leq\mu(G)\leq\mu(G^{2})=\left\vert M\right\vert $,
which implies $\mu(G)=\mu(G^{2})$.

According to Proposition \ref{th1}, it follows that $G$ is a square stable
K\"{o}nig-Egerv\'{a}ry graph having a perfect matching.
\end{proof}

Notice that the converse of Theorem \ref{th3} is not generally true; e.g.,
$G=C_{2n},n\geq2$. 

Now we are ready to formulate the main finding of the paper.

\begin{theorem}
\label{th4}For a graph $G$ of order $n\geq2$ the following assertions are equivalent:

\emph{(i)} $G^{2}$ is a K\"{o}nig-Egerv\'{a}ry graph;

\emph{(ii)} $G$ is a square-stable K\"{o}nig-Egerv\'{a}ry graph;

\emph{(iii)} $G$ has a perfect matching consisting of pendant edges;

\emph{(iv)} $G$ is very well-covered with exactly $\alpha(G)$ leaves.
\end{theorem}

\begin{proof}
The implication \emph{(i) }$\Longrightarrow$ \emph{(ii) }follows from Theorem
\ref{th3}. The proof of the implication \emph{(ii) }$\Longrightarrow$
\emph{(i)} is in the following series of inequalities:
\[
\left\vert V(G)\right\vert =\alpha(G)+\mu(G)=\alpha(G^{2})+\mu(G)\leq
\alpha(G^{2})+\mu(G^{2})\leq\left\vert V(G^{2})\right\vert =\left\vert
V(G)\right\vert .
\]

All the equivalences between \emph{(ii)},\emph{\ (iii) }and\emph{\ (iv)} have
been proved in \cite{LevMan2009}.\emph{\ }
\end{proof}

It was shown in \cite{Rav77} that a tree having at least two vertices is
well-covered if and only if it has a perfect matching consisting of pendant
edges. It was also mentioned there that every well-covered tree of order at
least two is very well-covered as well. Combining these observations with
Theorem \ref{th4} we obtain the following.

\begin{corollary}
The square of a tree is a K\"{o}nig-Egerv\'{a}ry graph if and only if the tree
is well-covered.
\end{corollary}

\section{Conclusions}

Recall that $\theta(G)$ is the \textit{clique covering number} of $G$, i.e.,
the minimum number of cliques whose union covers $V(G)$; $i(G)=\min
\{|S|:S\ $\textit{is a maximal stable set in }$G\}$, and $\gamma
(G)=\min\{|D|:D\ $\textit{is a minimal domination set in }$G\}$. In general,
it can be shown that the graph invariants mentioned above are related by the
following inequalities:
\[
\alpha(G^{2})\leq\theta(G^{2})\leq\gamma(G)\leq i(G)\leq\alpha(G)\leq
\theta(G)\text{, }%
\]
which turn out to be equalities, when $\alpha(G^{2})=\alpha(G)$ or
$\theta(G^{2})=\theta(G)$ \cite{ranvol}. 

It seems interesting to find out some other graph operations and invariants
such that interrelations between them may lead to K\"{o}nig-Egerv\'{a}ry graphs.

\end{document}